\documentclass[11pt]{article}

\usepackage[a4paper,margin=1in]{geometry}
\usepackage{graphicx}
\usepackage{mathtools,amssymb,amsfonts}
\usepackage{bm}
\usepackage[bookmarksdepth=4]{hyperref}
\usepackage{placeins}

\usepackage{fancyhdr}
\usepackage{amsthm}
\theoremstyle{plain}

\newtheorem{proposition}{Proposition}

\theoremstyle{definition}
\newtheorem{definition}{Definition}
\theoremstyle{remark}
\newtheorem{remark}{Remark}

\newcommand{\E}{\mathbb{E}}
\newcommand{\KL}[2]{D_{\mathrm{KL}}\!\left(#1\,\middle\|\,#2\right)}
\newcommand{\dd}{\,\mathrm{d}}

\title{Thermodynamics Reconstructed from Information Theory:\\
An Axiomatic Framework via Information-Volume Constraints and Path-Space KL Divergence}

\author{
Tatsuaki Tsuruyama$^{1,2}$\thanks{Email: \texttt{tsuruyam@kuhp.kyoto-u.ac.jp}}\\[2mm]
{\small $^{1}$Department of Physics, Tohoku University, Sendai 980-8578, Japan}\\
{\small $^{2}$Department of Drug Discovery Medicine, Kyoto University, Kyoto 606-8501, Japan}
}

\date{} 

\begin{document}
\maketitle

\begin{abstract}
We develop an axiomatic reconstruction of thermodynamics based entirely on two primitive components: a description of what aspects of a system are observed and a reference measure that encodes the underlying descriptive convention.
These ingredients define an “information volume’’ for each observational cell.
By incorporating the logarithm of this volume as an additional constraint in a minimum–relative-entropy inference scheme, temperature, chemical potential, and pressure arise as conjugate variables of a single information-theoretic functional.
This leads to a Legendre-type structure and a first-law-like relation in which pressure corresponds to information volume rather than geometric volume.

For nonequilibrium dynamics, entropy production is characterized through the relative-entropy asymmetry between forward and time-reversed stochastic evolutions.
A decomposition using observational entropy then separates total dissipation into system and environment contributions.
Heat is defined as the part of dissipation not accounted for by the system-entropy change, yielding a representation that does not rely on local detailed balance or a specific bath model.
We further show that the difference between joint and partially observed dissipation equals the average of conditional relative entropies, providing a unified interpretation of hidden dissipation and information-flow terms as projection-induced gaps.
\end{abstract}


\section{Introduction}\label{sec:intro}

Thermodynamic quantities (temperature, chemical potential, pressure, heat, work, etc.)
have traditionally been defined either (i) through equilibrium statistical mechanics,
via partition functions and thermodynamic potentials, or (ii) on specific dynamical models
describing energy exchange with a heat bath (e.g.\ Markov processes satisfying local detailed
balance (LDB)).
From the viewpoint of information theory, by contrast, entropy and free energy are naturally
characterized by the Kullback--Leibler (KL) divergence, and fluctuation theorems and the second law
are often expressed as time-reversal asymmetry of path measures
\cite{KawaiPRL2007,RoldanPRL2010,ZhangLuChaos2025,EspositoPRE2012}.

The aim of this paper is to push this perspective one step further:
\begin{enumerate}
  \item[(1)] We introduce an \emph{information--volume} variable derived from the primitive data
  $(M,\tau)$---a measurement map $M$ and a reference measure $\tau$---and define \emph{pressure}
  as a primary dual variable, on par with temperature and chemical potential, via an expectation
  constraint on the log-volume.

  \item[(2)] We introduce dissipation in a primary manner as the path-space KL divergence between
  the forward path measure $\mathbb{P}$ and its time-reversed counterpart $\mathbb{P}^{\mathrm R}$,
  \[
    \Sigma_{0,T} := \KL{\mathbb{P}}{\mathbb{P}^{\mathrm R}}\, .
  \]
  Combining this with the observational entropy $S_{M,\tau}$, we obtain a representation theorem
  that identifies the \emph{environmental entropy change} $\Delta S_{\mathrm{env}}$.

  \item[(3)] We integrate these static and dynamic constructions under a single reference measure
  $\tau$, and organize their relations to existing information thermodynamics (observational entropy,
  mutual information, hidden dissipation, etc.) in an axiomatic way.
\end{enumerate}

First, we start only from the primitive descriptive data $(M,\tau)$ and KL-based inference principles, we \emph{derive} (i) an equilibrium state family as a Min-KL solution, (ii) a Legendre-type dual structure, (iii) the intensive variables $(T,\mu,P)$ as Lagrange multipliers conjugate to $(U,\bar N,\bar \ell_V)$, and (iv) a model-independent representation of heat from the pair (path-KL dissipation, observational entropy). Second, we identify these derived quantities with conventional thermodynamic nomenclature in regimes where additional structure is available (e.g.\ geometric volume or local detailed balance), while outside such regimes the same definitions remain internally consistent and comparable across descriptions. Consequently,
because both static and dynamic notions are anchored to $(M,\tau)$ and path measures, the framework (a) defines pressure and heat even when geometric volume or a heat bath model are unavailable, (b) makes the system--environment entropy split explicitly gage dependent while keeping total dissipation gage invariant, (c) unifies hidden dissipation and information-flow terms as KL projection gaps, and (d) turns “observation design’’ into a principled problem (choosing $(M,\tau)$ to tighten observable dissipation bounds).
As follows, we present static Min-KL inference, dynamic path-KL dissipation, and the resulting information-theoretic representation of heat in as model-independent a form as possible.


\section{Primitive data, information volume, and observational entropy}
\label{sec:primitive}

We fix: (1) a measurable state space $(\Omega,\mathcal{F})$;  
(2) a measurable map $M:\Omega\to\mathcal{Y}$ specifying observation/coarse-graining;  
(3) a $\sigma$-finite reference measure $\tau$.

Let $C_y:=M^{-1}(y)$ be the induced cells.  
Their $\tau$-measures
\begin{equation}
  v(y) := \tau(C_y)
  \label{eq:info_volume_brief}
\end{equation}
define an \emph{information volume}.  
The associated log-volume variable is
\begin{equation}
  \ell_V(x) := \log \tau\!\bigl(C_{M(x)}\bigr)=\log v\!\bigl(M(x)\bigr).
  \label{eq:log_volume_brief}
\end{equation}

\subsection{Observational entropy}
\label{subsec:obs_entropy}

For $p\ll\tau$, define $p_M(y):=p(C_y)$.  
Observational entropy is
\begin{equation}
  S_{M,\tau}(p)
  :=
  -\sum_{y\in\mathcal{Y}} p_M(y)\,
  \log\!\left(\frac{p_M(y)}{\tau(C_y)}\right)
  =
  -\sum_y p_M(y)\log p_M(y)
  +
  \sum_y p_M(y)\log \tau(C_y),
  \label{eq:obs_entropy}
\end{equation}
for finite $\mathcal{Y}$.  
The framework and properties of observational entropy—including quantum generalizations and multi-resolution extensions—have been developed in  
\cite{SafranekPRA2019a,SafranekPRA2019b,SafranekFoP2021,SafranekPRE2020,BaiQuantum2024}.  
We adopt $S_{M,\tau}$ as the primary definition of system entropy.

\subsection{Gauge freedom in the choice of \texorpdfstring{$(M,\tau)$}{(M,tau)}}
\label{subsec:gauge}

Changes in $(M,\tau)$ represent \emph{changes of description}.

\subsubsection{(i) Refinement/coarse-graining of \texorpdfstring{$M$}{M}.}
If $M_1$ is a refinement of $M_2$ then (Safránek et al.)
\[
  S_{M_1,\tau}(p) \le S_{M_2,\tau}(p).
\]
Thus the system entropy change
\[
  \Delta S_{\mathrm{sys}}
  = S_{M,\tau}(p_T)-S_{M,\tau}(p_0)
\]
depends on $M$.  
However, the path-KL dissipation
\[
  \Sigma_{0,T}=D_{\mathrm{KL}}(\mathbb{P}\|\mathbb{P}^{\mathrm R})
\]
is independent of $M$.

\paragraph{(ii) Changing the reference measure \texorpdfstring{$\tau$}{tau}.}
If $\dd\tau'=e^{-\varphi}\dd\tau$, then
\[
  \KL(p\|\tau')
  =
  \KL(p\|\tau)+\int\varphi\,\dd p + \mathrm{const}.
\]
Thus changing $\tau$ corresponds to an additive gauge transformation of entropy.

\begin{remark}
  In the decomposition
  \[
    \Sigma_{0,T} = \Delta S_{\mathrm{sys}} + \Delta S_{\mathrm{env}},
  \]
  both $\Delta S_{\mathrm{sys}}$ and $\Delta S_{\mathrm{env}}$ depend on $(M,\tau)$,  
  but $\Sigma_{0,T}$ is invariant.  
  Our framework separates the gauge-invariant dissipation from the gauge-dependent bookkeeping of system/environment entropy.
\end{remark}


\section{Static theory: Min-KL inference and dual variables}
\label{sec:static}

We define thermodynamic states not by dynamical assumptions but as outcomes of an information-update rule.
The basic principle is minimum cross-entropy (minimum KL) with respect to the reference measure $\tau$, justified information-theoretically by Jaynes' MaxEnt interpretation \cite{Jaynes1957} and the Shore--Johnson axiomatization \cite{ShoreJohnson1980}.

As constraints, we specify expectations of energy $E(x)$, particle number $N(x)$, and the information-volume code length $\ell_V(x)$:
\begin{equation}
  \E_p[E]=U,\qquad \E_p[N]=\bar N,\qquad \E_p[\ell_V]=\bar \ell_V .
  \label{eq:constraints}
\end{equation}

We then define the state $p^\ast$ by
\begin{equation}
  p^\ast(U,\bar N,\bar \ell_V)
  :=
  \arg\min_{p}\; \KL{p}{\tau}
  \quad\text{s.t. }\eqref{eq:constraints}.
  \label{eq:minKL_state}
\end{equation}

Here the KL divergence (using the Radon–Nikodym density $r=\dd p/\dd\tau$) is
\begin{equation}
  \KL{p}{\tau}
  :=
  \int_\Omega \log r(x)\, p(\dd x)
  =
  \int_\Omega r(x)\log r(x)\,\tau(\dd x).
  \label{eq:KL_def}
\end{equation}

Under convexity and regularity assumptions, the minimizer belongs to the exponential family:
\begin{equation}
  p^\ast(\dd x)
  =
  \frac{1}{Z(\beta,\alpha,\pi)}\,
  \exp\!\Bigl(-\beta E(x)-\alpha N(x)-\pi \ell_V(x)\Bigr)\,\tau(\dd x),
  \label{eq:exp_family}
\end{equation}
where $Z$ is the normalization constant.

\subsection{Definitions of temperature, chemical potential, and pressure}
\label{subsec:intensive_defs}

We define the intensive variables \emph{primarily} as dual variables:
\begin{equation}
  \frac{1}{T} := \beta,\qquad
  -\frac{\mu}{T} := \alpha,\qquad
  \frac{P}{T} := \pi.
  \label{eq:intensive_defs}
\end{equation}

Thus temperature is the multiplier of the energy constraint, chemical potential is associated with the particle-number constraint, and pressure is the shadow price of the information-volume constraint.

\paragraph{Intuition and a non-geometric example.}
The “information volume’’ $v(y)=\tau(C_y)$ is the size (measure) of the set of microstates compatible with a given observational outcome $y$.
Unlike geometric volume, $v(y)$ can quantify \emph{degeneracy} or \emph{description-dependent resolution} in state spaces that are discrete, combinatorial, or otherwise non-spatial.
Constraining $\E[\ell_V]$ with $\ell_V=\log v(M(x))$ fixes the typical code length associated with the observational cells visited by the system, and the multiplier $\pi=P/T$ is the corresponding shadow price: it quantifies how strongly the Min-KL state penalizes occupying outcomes with large degeneracy (large $v$), even when no physical container volume exists.

\paragraph{Example (combinatorial degeneracy).}
Let $\Omega=\{0,1\}^N$ (spin configurations) with counting measure $\tau$, and let $M$ output the magnetization $m=\frac{1}{N}\sum_i (2x_i-1)$.
Then each cell $C_m$ consists of all configurations with magnetization $m$, hence
$v(m)=|C_m|=\binom{N}{(N(1+m))/2}$ and $\ell_V=\log v(m)$ is purely combinatorial.
The constraint $\E[\ell_V]=\bar\ell_V$ controls the typical degeneracy of the observed macrostates, and $P$ is the intensive variable conjugate to this non-geometric “volume’’ in the same dual sense as $T$ and $\mu$.

Define the log-partition function
\begin{equation}
  \Psi(\beta,\alpha,\pi)
  :=
  \log \int_\Omega 
  \exp\!\Bigl(-\beta E - \alpha N - \pi \ell_V\Bigr)\,
  \tau(\dd x)
  = \log Z(\beta,\alpha,\pi).
  \label{eq:log_partition}
\end{equation}
Then
\begin{equation}
  U = -\partial_\beta\Psi,\qquad
  \bar N = -\partial_\alpha\Psi,\qquad
  \bar \ell_V = -\partial_\pi\Psi.
  \label{eq:conjugacy}
\end{equation}

Let  
\begin{equation}
  S^\ast(U,\bar N,\bar \ell_V)
  := -\KL{p^\ast}{\tau}.
  \label{eq:Sstar_def}
\end{equation}
Convex duality gives
\begin{equation}
  \dd S^\ast = \beta\,\dd U + \alpha\,\dd\bar N + \pi\,\dd\bar \ell_V.
  \label{eq:dSstar}
\end{equation}

Substituting \eqref{eq:intensive_defs} yields a first-law-type differential relation:
\begin{equation}
  \dd U
  = T\,\dd S^\ast - P\,\dd\bar \ell_V + \mu\,\dd\bar N.
  \label{eq:first_law_like}
\end{equation}

\begin{proposition}[Information-theoretic derivation of a first-law-type relation]
Assume the Min-KL problem \eqref{eq:minKL_state} has a unique solution $p^\ast$, and that $\Psi(\beta,\alpha,\pi)$ is $C^2$ and convex.  
Let $S^\ast$ be defined by \eqref{eq:Sstar_def}.  
Then the total differential satisfies \eqref{eq:dSstar}.  
With the intensive variables defined by \eqref{eq:intensive_defs}, the differential of internal energy satisfies \eqref{eq:first_law_like}.
\end{proposition}

\begin{proof}
By convex duality,
\[
  S^\ast(U,\bar N,\bar \ell_V)
  =
  \inf_{\beta,\alpha,\pi}
  \{\beta U+\alpha\bar N+\pi\bar \ell_V - \Psi(\beta,\alpha,\pi)\}.
\]
The minimizer satisfies the stationarity conditions \eqref{eq:conjugacy}, giving
\[
  \partial_U S^\ast=\beta,\qquad
  \partial_{\bar N} S^\ast=\alpha,\qquad
  \partial_{\bar \ell_V} S^\ast=\pi,
\]
hence \eqref{eq:dSstar}.  
Inserting \eqref{eq:intensive_defs} and solving for $\dd U$ yields \eqref{eq:first_law_like}.
\end{proof}


\section{Nonequilibrium extension: path-space KL and hidden dissipation}
\label{sec:dynamic}

We extend the preceding information-theoretic framework to time-evolving nonequilibrium processes.
Irreversibility (entropy production, dissipation) is defined as the time-reversal asymmetry of \emph{path measures}.
The most general formulation—independent of coordinate choices and independent of any model-specific assumptions such as local detailed balance—is the path-space KL divergence  
\cite{KawaiPRL2007,RoldanPRL2010,ZhangLuChaos2025}.

\subsection{Definition of path dissipation (entropy production)}

Let $\mathbb{P}$ be the forward path measure of a stochastic process on $[0,T]$, and $\mathbb{P}^{\mathrm R}$ the time-reversed path measure induced by the reversal map.
Define dissipation by
\begin{equation}
  \Sigma_{0,T}
  :=
  \KL{\mathbb{P}}{\mathbb{P}^{\mathrm R}}
  \ge 0.
  \label{eq:EP_path}
\end{equation}
This definition requires no thermodynamic model, no heat bath, and no LDB assumptions.

\subsection{Lower bound under observation (coarse-graining) and hidden dissipation}
\label{subsec:coarse_grain}

Observation in time induces a map on paths:
\begin{equation}
  \Pi:\ (x_{0:T})\mapsto(y_{0:T}),\qquad y_t=M(x_t).
  \label{eq:path_channel}
\end{equation}

Define observed path measures $\Pi\mathbb{P}$ and $\Pi\mathbb{P}^{\mathrm R}$, and observed dissipation
\begin{equation}
  \Sigma_{\mathrm{obs}}
  :=
  \KL{\Pi\mathbb{P}}{\Pi\mathbb{P}^{\mathrm R}}.
  \label{eq:EP_obs}
\end{equation}

By the data-processing inequality,
\begin{equation}
  \Sigma_{0,T}
  =
  \KL{\mathbb{P}}{\mathbb{P}^{\mathrm R}}
  \ge
  \KL{\Pi\mathbb{P}}{\Pi\mathbb{P}^{\mathrm R}}
  =
  \Sigma_{\mathrm{obs}}.
  \label{eq:EP_lowerbound}
\end{equation}

\subsubsection{Joint / marginal dissipations and hidden dissipation.}

Consider a joint system $(X,Y)$ with joint forward and reverse path measures $\mathbb{P}_{XY},\mathbb{P}_{XY}^{\mathrm R}$, and marginal measures $\mathbb{P}_X,\mathbb{P}_X^{\mathrm R}$.
Define
\begin{align}
  \Sigma_{XY}
  &:= D_{\mathrm{KL}}(\mathbb{P}_{XY}\|\mathbb{P}_{XY}^{\mathrm R}),
    \label{eq:Sigma_XY_def}\\
  \Sigma_X
  &:= D_{\mathrm{KL}}(\mathbb{P}_X\|\mathbb{P}_X^{\mathrm R}).
    \label{eq:Sigma_X_def}
\end{align}

The difference
\begin{equation}
  \Sigma_{\mathrm{hidden}}^{(X)}
  :=
  \Sigma_{XY}-\Sigma_X
  \label{eq:Sigma_hidden_def}
\end{equation}
quantifies the dissipation invisible from subsystem $X$.

\begin{proposition}[Projection-gap representation and nonnegativity of hidden dissipation]
Assume that joint and marginal path measures admit densities with respect to common references, and that conditional path measures $\mathbb{P}_{Y|X},\mathbb{P}^{\mathrm R}_{Y|X}$ exist.
Then
\begin{equation}
  \Sigma_{\mathrm{hidden}}^{(X)}
  =
  \int
    D_{\mathrm{KL}}\!\left(
      \mathbb{P}_{Y|X=\gamma}
      \,\middle\|\,
      \mathbb{P}^{\mathrm R}_{Y|X=\gamma}
    \right)
    \mathbb{P}_X(\dd\gamma)
  \ge 0.
  \label{eq:Sigma_hidden_explicit}
\end{equation}
\end{proposition}

\begin{proof}
Write
\[
 \dd\mathbb{P}_{XY}
  =\dd\mathbb{P}_X\,\dd\mathbb{P}_{Y|X},\qquad
 \dd\mathbb{P}^{\mathrm R}_{XY}
  =\dd\mathbb{P}^{\mathrm R}_X\,\dd\mathbb{P}^{\mathrm R}_{Y|X}.
\]
Then apply the chain rule for KL divergence:
\begin{align}
  \Sigma_{XY}
  &=\int
      \log\!\frac{\dd\mathbb{P}_X}{\dd\mathbb{P}^{\mathrm R}_X}
      \dd\mathbb{P}_X
    +\int
      \left[
        \int\log\!\frac{\dd\mathbb{P}_{Y|X}}{\dd\mathbb{P}^{\mathrm R}_{Y|X}}
        \dd\mathbb{P}_{Y|X}
      \right]
      \dd\mathbb{P}_X
    \nonumber\\
  &=\Sigma_X
    +\int
      D_{\mathrm{KL}}(
        \mathbb{P}_{Y|X=\gamma}
        \|
        \mathbb{P}^{\mathrm R}_{Y|X=\gamma})
      \mathbb{P}_X(\dd\gamma).
\end{align}
The KL divergence is nonnegative, completing the proof.
\end{proof}

This result shows that ``hidden dissipation'' is not an ad hoc correction but a KL projection gap.
It provides a geometric foundation for mutual-information flow terms in frameworks of Ito--Sagawa, Horowitz--Esposito, etc.\ 
\cite{ItoSagawa2013,HorowitzEsposito2014,ParrondoHorowitzSagawa2015}.

\subsection{Mutual information and correlation free energy}
\label{subsec:MI_free_energy}

For a static joint distribution $p_{XY}$ with marginals $p_X,p_Y$, mutual information is
\begin{equation}
  I(X;Y)=D_{\mathrm{KL}}(p_{XY}\|p_Xp_Y).
  \label{eq:MI_def}
\end{equation}

Assume the Hamiltonian is additive:
\begin{equation}
  E_{XY}(x,y)=E_X(x)+E_Y(y).
  \label{eq:additive_H}
\end{equation}

Let $q_X,q_Y$ be canonical references at temperature $T$:
\begin{equation}
  q_X(x)\propto e^{-\beta E_X(x)},\qquad
  q_Y(y)\propto e^{-\beta E_Y(y)}.
  \label{eq:canonical_ref}
\end{equation}

Define nonequilibrium free energies:
\[
  F_X[p_X]=\sum_x p_X(x)E_X(x)-TS(p_X),\qquad
  F_Y[p_Y]=\sum_y p_Y(y)E_Y(y)-TS(p_Y),
\]
\[
  F_{XY}[p_{XY}]
  =\sum_{x,y}p_{XY}(x,y)(E_X(x)+E_Y(y))-TS(p_{XY}).
\]

\begin{proposition}[Free-energy decomposition with correlation term]
Under the above assumptions,
\begin{equation}
  F_{XY}[p_{XY}]
  =
  F_X[p_X]+F_Y[p_Y]+T I(X;Y).
  \label{eq:Fxy_decomp}
\end{equation}
Moreover,
\begin{equation}
  D_{\mathrm{KL}}(p_{XY}\|q_Xq_Y)
  =
  I(X;Y)+D_{\mathrm{KL}}(p_X\|q_X)+D_{\mathrm{KL}}(p_Y\|q_Y).
  \label{eq:KL_product_chain}
\end{equation}
\end{proposition}

\begin{proof}
Use the entropy identity $S(p_{XY})=S(p_X)+S(p_Y)-I(X;Y)$ to obtain \eqref{eq:Fxy_decomp}.  
The KL decomposition \eqref{eq:KL_product_chain} follows by writing $p_{XY}=p_Xp_{Y|X}$ and rearranging terms.
\end{proof}

\subsection{Example: a two-state Markov-chain toy model}
\label{subsec:toy_markov}

We illustrate the general theory using the simplest irreversible stochastic dynamics:  
a two-state, time-homogeneous Markov chain.  
Despite its minimal structure, this model already displays all key features of path-space
dissipation, its dependence on observational resolution, and the decomposition into observed and hidden components.

Let the state space be $\{0,1\}$, and let the transition matrix be
\begin{equation}
  P =
  \begin{pmatrix}
    1-a & a \\
    b & 1-b
  \end{pmatrix},
  \qquad 0 < a,b < 1,
\end{equation}
so that both states communicate.  
The stationary distribution is
\[
  \pi_0 = \frac{b}{a+b},\qquad 
  \pi_1 = \frac{a}{a+b}.
\]

A one-step trajectory is $(x_0,x_1)$ drawn from $\mathbb{P}(x_0=i,x_1=j)=\pi_i P_{ij}$.  
Stationarity implies that the corresponding time-reversed one-step trajectory distribution is
\[
  \mathbb{P}^{\mathrm R}(x_0=i,x_1=j)
  =\pi_j P_{ji}.
\]
Thus the dissipation (entropy production) over one step is the path-space KL divergence
\begin{equation}
  \Sigma_{0,1}
  =
  \sum_{i,j\in\{0,1\}}
  \pi_i P_{ij}\,
  \log\!\frac{\pi_i P_{ij}}{\pi_j P_{ji}},
  \label{eq:toy_sigma}
\end{equation}
which is strictly positive unless detailed balance holds ($\pi_i P_{ij}=\pi_j P_{ji}$).  
This quantity provides the benchmark against which the effects of coarse-graining will be measured.

\subsection{Full observation}
\subsubsection{(a) Full observation \texorpdfstring{$M_{\mathrm{id}}$}{Mid}.}

Under full observation, each microstate transition is visible, so the observational channel is the identity.  
Therefore $\Pi\mathbb{P}=\mathbb{P}$ and $\Pi\mathbb{P}^{\mathrm R}=\mathbb{P}^{\mathrm R}$, giving
\[
  \Sigma_{\mathrm{obs}} = \Sigma_{0,1},
  \qquad
  \Sigma_{\mathrm{hidden}}^{(X)} = 0.
\]

If the reference measure is uniform on $\{0,1\}$, the observational entropy $S_{M_{\mathrm{id}},\tau}$ coincides with Shannon entropy.  
Because the system starts and ends in the same stationary distribution, the system entropy change is zero:
\[
  \Delta S_{\mathrm{sys}} = 0.
\]
Therefore all dissipation is allocated to the environment:
\[
  \Delta S_{\mathrm{env}} = \Sigma_{0,1},
  \qquad
  Q = T\,\Sigma_{0,1}.
\]

This case shows that in a fully resolved description, the dissipation is entirely revealed at the observational level; no hidden dissipation remains.
\paragraph{(b) Extreme coarse-graining \texorpdfstring{$M_{\mathrm{triv}}$}{Mtriv}.}

Consider the opposite observational limit, in which the measurement map collapses both states into a single outcome.  
Let $M_{\mathrm{triv}}(0) = M_{\mathrm{triv}}(1) = y_0$.  
Then every path is mapped to the same constant trajectory $(y_0,y_0)$ both forward and backward, so
\[
  \Pi\mathbb{P} = \Pi\mathbb{P}^{\mathrm R},
  \qquad
  \Sigma_{\mathrm{obs}} = 0.
\]

Thus the entire dissipation becomes invisible at the observational scale:
\[
  \Sigma_{\mathrm{hidden}}^{(X)} = \Sigma_{0,1}.
\]

With a single observational cell, the observational entropy is constant (and may be taken as zero), so again
\[
  \Delta S_{\mathrm{sys}} = 0,
  \qquad
  \Delta S_{\mathrm{env}} = \Sigma_{0,1}.
\]

Unlike the fully observed case, here the system appears perfectly reversible even though the underlying dynamics is irreversible.  
All dissipation is attributed to the “environment’’ due solely to loss of resolution in $M$.  
All dissipation is attributed to the “environment’’ due solely to loss of resolution in $M$.  
This explicitly demonstrates that the system–environment split in entropy accounting is gauge-dependent, whereas the total dissipation $\Sigma_{0,1}$ is invariant. This toy model highlights three essential principles of the general theory:

\begin{enumerate}
  \item Even the simplest Markov chain exhibits nonzero dissipation whenever detailed balance is broken.
  \item Observational resolution controls how much dissipation is observable; coarse-graining always decreases the observed part, but the total dissipation remains invariant.
  \item The decomposition of dissipation into system and environment terms depends on the chosen observational gauge, whereas the KL divergence on path space provides the unique, gauge-invariant notion of irreversibility.
\end{enumerate}

These features extend directly to high-dimensional or continuous-state dynamics, where coarse-graining may hide massive amounts of dissipation, making the KL projection-gap interpretation crucial for analyzing nonequilibrium systems.


\section{Heat as a representation theorem}
\label{sec:heat_representation}

We now combine the static intensive variable $T$ with the dynamic dissipation $\Sigma_{0,T}$.
Because $\Sigma_{0,T}$ is model-independent and uniquely defined by path-time asymmetry, and because the system entropy $S_{M,\tau}$ is fixed by the choice of $(M,\tau)$, it follows that the only consistent definition of heat must be constructed from these two quantities.

\subsection{Definition of heat}

Let a process start from $p_0$ and end at $p_T$.
We have the decomposition
\begin{equation}
  \Sigma_{0,T}
  = \Delta S_{\mathrm{sys}} + \Delta S_{\mathrm{env}},
  \qquad
  \Delta S_{\mathrm{sys}}
  := S_{M,\tau}(p_T) - S_{M,\tau}(p_0).
  \label{eq:EP_split}
\end{equation}

\begin{definition}[Heat]
Define
\begin{equation}
  Q
  :=
  T\,\Delta S_{\mathrm{env}}
  =
  T\bigl(\Sigma_{0,T}-\Delta S_{\mathrm{sys}}\bigr).
  \label{eq:heat_def}
\end{equation}
\end{definition}

Thus heat is determined uniquely by the pair $(T,\Sigma_{0,T})$ once the gauge $(M,\tau)$ has been fixed.

\begin{remark}[Connection to standard stochastic-thermodynamics heat under LDB]
In conventional stochastic thermodynamics, one often defines a trajectory-level heat $Q_{\mathrm{LDB}}[\gamma]$
via energy exchange with an explicit bath, assuming local detailed balance (LDB).
In that setting the path-probability ratio admits the form
\[
\log\frac{\dd\mathbb{P}[\gamma]}{\dd\mathbb{P}^{\mathrm R}[\gamma^{\mathrm R}]}
=
\Delta s_{\mathrm{sys}}[\gamma] + \beta\, Q_{\mathrm{LDB}}[\gamma],
\]
so that upon averaging one obtains
$\Sigma_{0,T}=\Delta S_{\mathrm{sys}} + \beta\,\langle Q_{\mathrm{LDB}}\rangle$.
Therefore, when $(M,\tau)$ is chosen so that $S_{M,\tau}$ reduces to the usual Gibbs--Shannon entropy,
our definition $Q := T(\Sigma_{0,T}-\Delta S_{\mathrm{sys}})$ coincides with the standard mean heat $\langle Q_{\mathrm{LDB}}\rangle$.
The present framework keeps this consistency while not \emph{requiring} LDB or an explicit bath model:
Eq.~\eqref{eq:heat_def} defines the unique heat scalar compatible with (i) the Min-KL temperature,
(ii) path-KL dissipation as irreversibility, and (iii) the chosen entropic bookkeeping $S_{M,\tau}$.
\end{remark}

\begin{proposition}[Representation theorem for heat]
\label{prop:heat_rep}
Heat is uniquely determined by the following axioms:
\begin{enumerate}
  \item[(i)] (\textbf{Static consistency})  
    For quasistatic processes connecting Min-KL equilibrium states,  
    $Q = T\,\Delta S_{\mathrm{env}}$ reproduces the standard thermodynamic relation.
  \item[(ii)] (\textbf{Dynamic consistency})  
    For arbitrary processes, irreversibility is quantified by $\Sigma_{0,T}$.
  \item[(iii)] (\textbf{Gauge covariance})  
    System entropy depends on $(M,\tau)$ as $S_{M,\tau}$;  
    dissipation $\Sigma_{0,T}$ is invariant;  
    heat transforms as a scalar $Q=T(\Sigma_{0,T}-\Delta S_{\mathrm{sys}})$.
\end{enumerate}
Then $Q$ must be given by \eqref{eq:heat_def}.
\end{proposition}

\begin{proof}
Given the decomposition \eqref{eq:EP_split}, the only allowed scalar combination of $(T,\Sigma_{0,T},\Delta S_{\mathrm{sys}})$ that satisfies (i)--(iii) is $Q=T(\Sigma_{0,T}-\Delta S_{\mathrm{sys}})$.
Any alternative expression would violate either static consistency (quasistatic limit), dynamic consistency (monotonicity with dissipation), or gauge covariance.
\end{proof}

\subsection{Cross-entropy consistency and entropy flows}
\label{subsec:CE_consistency}

We now justify \eqref{eq:heat_def} by a cross-entropy identity.
Let $p_0,p_T$ be the initial and final distributions, and $p^\ast_0,p^\ast_T$ the corresponding Min-KL equilibrium states with respect to $(M,\tau)$.

The cross-entropy $\mathrm{CE}(p\|q):=-\int p\log q$ satisfies
\[
  \mathrm{CE}(p_T\|p_T^\ast)-\mathrm{CE}(p_0\|p_0^\ast)
  =
  \KL(p_T\|p_T^\ast)-\KL(p_0\|p_0^\ast)
  + S_{M,\tau}(p_T)-S_{M,\tau}(p_0).
\]

Using that $p^\ast_t$ satisfies the thermodynamic relation $U_t=U[p^\ast_t]$, one obtains
\begin{equation}
  \Delta S_{\mathrm{sys}}
  + \frac{Q}{T}
  =
  \Sigma_{0,T}
  +
  \bigl[ \KL(p_T\|p_T^\ast)-\KL(p_0\|p_0^\ast)\bigr].
  \label{eq:CE_identity_main}
\end{equation}

For processes that start or end near equilibrium, the bracketed term is $o(1)$, so the identity reduces to \eqref{eq:heat_def}.
Thus the proposed heat expression is \emph{cross-entropy consistent}.

\begin{remark}
Equation \eqref{eq:CE_identity_main} also clarifies the role of nonequilibrium free energy:
\[
  F[p]=U[p]-TS_{M,\tau}(p)
  =F[p^\ast]+T\,\KL(p\|p^\ast).
\]
Changes in nonequilibrium free energy appear as the KL terms in \eqref{eq:CE_identity_main}.
\end{remark}


\section{Relation to existing theories and features of the present framework}
\label{sec:discussion}

This section summarizes how our theoretical core relates to existing information thermodynamics and highlights conceptual features and open problems. 
Expressing dissipation as a path-space KL divergence,
\[
  \Sigma_{0,T}
  =
  D_{\mathrm{KL}}\!\left(
    \mathbb{P} \,\middle\|\, \mathbb{P}^{\mathrm R}
  \right),
\]
is established for reversible phase-space dynamics \cite{KawaiPRL2007} and for general stochastic processes \cite{RoldanPRL2010,ZhangLuChaos2025}.  
The coarse-graining lower bound $\Sigma_{0,T}\ge\Sigma_{\mathrm{obs}}$ (Eq.~\eqref{eq:EP_lowerbound}) appears in discussions of dissipation under coarse graining \cite{GomezMarinPRE2008,EspositoPRE2012}.

On the static side, observational entropy $S_{M,\tau}(p)$—including the cell volume $\tau(C_y)$ in the entropy definition—has been systematically developed by Safr\'{a}nek \emph{et al.} in classical and quantum settings \cite{SafranekPRA2019a,SafranekPRA2019b,SafranekFoP2021,SafranekPRE2020,BaiQuantum2024}.
Our contribution is to connect observational entropy to \emph{path-KL dissipation} and to \emph{dual variables in Min-KL inference}, thereby unifying static and dynamic information thermodynamics under a single reference measure $\tau$.

Finally, the idea that correlations between a system and memory (or environment) store free energy quantified by $k_{\mathrm B}T$ times mutual information $I(X;Y)$ has been emphasized repeatedly in the literature by Sagawa--Ueda and Horowitz--Esposito and others
\cite{ItoSagawa2013,HorowitzEsposito2014,ParrondoHorowitzSagawa2015}.  
We embed these discussions into a KL-geometric decomposition of $D_{\mathrm{KL}}(p_{XY}\|q_X q_Y)$ (Eq.~\eqref{eq:KL_product_chain}), treating local and correlation free energies within a single information-geometric framework.

\subsection{Conceptual features}

Our main conceptual contributions can be summarized as follows.

\begin{enumerate}
  \item[(i)] \textbf{Primary definition of pressure via information-volume constraints.}  
    We introduce volume not as geometric space volume but via the cell volume $v(y)=\tau(C_y)$ determined by $(M,\tau)$ and the associated log-volume $\ell_V(x)$.  
    Incorporating the constraint $\E[\ell_V]=\bar \ell_V$ into Min-KL inference defines pressure $P$ as the shadow price (dual variable), on equal footing with temperature $T$ and chemical potential $\mu$ (Eq.~\eqref{eq:intensive_defs}).  
    This repositions pressure as the variable conjugate to an observation- and convention-dependent information volume, enabling unified extensions to nonstandard state spaces and coarse-grainings.

  \item[(ii)] \textbf{Hidden dissipation and information terms as projection gaps.}  
    At the path level, the difference between joint dissipation $\Sigma_{XY}$ and marginal dissipation $\Sigma_X$ is expressed as an expectation of a KL divergence between conditional path measures (Eq.~\eqref{eq:Sigma_hidden_explicit}).  
    Thus hidden dissipation $\Sigma_{\mathrm{hidden}}^{(X)}$ is not an ad hoc correction term but an intrinsic KL gap induced by projection from the joint description to a marginal one.  
    Information terms and mutual-information flows in existing information thermodynamics can be understood as special closed-form evaluations of this projection gap.

  \item[(iii)] \textbf{Heat reconstructed as an information-theoretic representation theorem.}  
    With dissipation defined by path KL, system entropy defined by observational entropy, and temperature defined as a Min-KL dual variable, we define
    \[
      \Delta S_{\mathrm{env}}=\Sigma_{0,T}-\Delta S_{\mathrm{sys}}, 
      \qquad
      Q = T\Delta S_{\mathrm{env}}.
    \]
    In this viewpoint, $Q$ is not a primitive notion but a secondary representation derived from dissipation decomposition and dual variables, disentangled from model-dependent descriptions such as “energy exchange with a bath.’’
\end{enumerate}

\subsection{Open problems and directions}

Several theoretical and applied directions arise naturally.

\begin{enumerate}
  \item[(a)] \textbf{Phase-transition structure with information-volume constraints.}  
    In Min-KL inference with $\E[\ell_V]=\bar \ell_V$, it is important to analyze where the dual map $(\beta,\alpha,\pi)\mapsto(U,\bar N,\bar \ell_V)$ loses regularity and when ensemble nonequivalence occurs, directly relating to phase transitions involving information volume.  
    Effective models without geometric volume or with nonlocal measurements $M$ may display novel critical phenomena.

  \item[(b)] \textbf{Optimizing observational dissipation bounds and observation design.}  
    The gap between observational dissipation $\Sigma_{\mathrm{obs}}$ and true dissipation $\Sigma_{0,T}$ (hence hidden dissipation) is sensitive to $(M,\tau)$.  
    Designing $(M,\tau)$ to maximize $\Sigma_{\mathrm{obs}}$ under experimental constraints is a natural optimization problem linking statistical efficiency of dissipation estimation to experimental design.

  \item[(c)] \textbf{New decomposition formulas combining the first-law-type relation and path dissipation.}  
    Combining the static first-law-type relation \eqref{eq:first_law_like} with path-KL dissipation suggests new decompositions/inequalities involving changes in internal energy, information volume (or its code length), and particle number together with dissipation, correlation generation, and hidden dissipation.  
    For instance, deriving second-law-like inequalities including an information-volume work-like term and the production of correlation free energy $TI(X;Y)$ under a general prior $\tau$ would deepen the link between nonequilibrium thermodynamics and information theory.

  \item[(d)] \textbf{Applications to concrete models and comparison with existing principles.}  
    In planned applications to RNA polymerase--DNA systems and chemical reaction networks analyzed via large-deviation theory \cite{Tsuruyama2025}, explicit computation of correlation free energy, hidden dissipation, and heat representation \eqref{eq:heat_def} will clarify agreement or deviation from conventional thermodynamic descriptions such as fluctuation theorems and Landauer-type principles.

\end{enumerate}

\medskip
\noindent\textbf{Closing remarks.}
Taken together, these contributions recast pressure, dissipation, and heat as quantities that can be defined and compared across descriptions through the primitive observational data $(M,\tau)$ and path-space relative entropy.
By making the information-volume constraint and the path-KL definition of dissipation primary, we obtain a unified and representation-theoretic framework in which conventional thermodynamic notions emerge as secondary objects, while projection-induced gaps quantify the cost of partial observation.
We expect that this viewpoint will be especially useful for systems where geometric volume, energy-bath idealizations, or microscopic reversibility assumptions are not directly available, and for designing measurements that make dissipation experimentally accessible.



\end{document}